\documentclass[11pt]{article}

\usepackage[utf8]{inputenc}
\usepackage[T1]{fontenc}
\usepackage{amsmath,amssymb,amsthm}
\usepackage{graphicx}
\usepackage{hyperref}
\usepackage{geometry}
\usepackage{bbm}
\usepackage{amsmath}
\usepackage{amsthm}
\usepackage{hyperref}
\usepackage{cleveref}
\usepackage{nicefrac}
\usepackage{thmtools,thm-restate}
\usepackage{mathabx}

\usepackage{tcolorbox}

\usepackage{hhline}
\usepackage{enumitem}
\usepackage{array}
\usepackage{framed}
\usepackage{multicol}

\usepackage{algpseudocode}
\usepackage{algorithm}

\usepackage[suppress]{color-edits}
\addauthor[sz]{sz}{red}
\addauthor[lb]{lb}{blue}
\usepackage{upgreek}
\usepackage{physics}

\newcommand{\cD}{\mathcal{D}}

\newcommand{\bbR}{\mathbb{R}}

\newcommand{\Ex}{\mathbb{E}}

\newcommand{\eps}{\varepsilon}

\newcommand{\down}[1]{\left\lfloor #1\right\rfloor}

\newcommand{\floor}[1]{\left\lfloor #1\right\rfloor}

\newcommand{\IGNORE}[1]{}

\newcommand{\A}{\mathsf{A}}

\newcommand{\ind}[1]{\mathbbm{1}\{#1\}}

\newcommand{\hatf}{\hat{f}}

\newcommand{\hatr}{\hat{r}}

\newcommand{\cA}{\mathcal{A}}
\newcommand{\cM}{\mathcal{M}}
\newcommand{\cL}{\mathcal{L}}

\newcommand{\vecs}{{s}}

\newcommand{\sfA}{\mathsf{A}}

\newcommand{\Ag}{\mathsf{A}}

\newcommand{\pc}{\mathsf{P}}
\renewcommand{\P}{\mathsf{P}}

\newcommand{\cY}{\mathcal{Y}}

\newcommand{\OPT}{\mathtt{OPT}}

\newcommand{\APX}{\mathtt{APX}}

\newcommand{\LP}{\mathtt{LP}}

\newcommand{\hatc}{\hat{c}}

\newcommand{\C}{\mathcal{C}}

\newcommand{\hatG}{\hat{G}}
\newcommand{\Y}{\mathcal{Y}}

\newcommand{\Deltas}{\Delta s}

\newcommand{\UCB}{\mathtt{UCB}}
\newcommand{\LCB}{\mathtt{LCB}}

\renewcommand{\bar}[1]{\widebar{#1}}

\newcommand{\Reg}{\mathtt{Reg}}

\renewcommand{\i}{ {(i)} }
\renewcommand{\j}{ {(j)} }
\renewcommand{\k}{ {(k)} }

\newcommand{\paren}[1]{\left( #1 \right)}

\renewcommand{\tilde}[1]{\widetilde{#1}}

\let\alpha\upalpha

\newtheorem{theorem}{Theorem}
\newtheorem{lemma}{Lemma}
\newtheorem{proposition}{Proposition}
\newtheorem{corollary}{Corollary}
\newtheorem{remark}{Remark}
\newtheorem{assumption}{Assumption}

\newtheorem{definition}{Definition}

\newtheorem*{lemma*}{Lemma}

\newcommand{\set}[1]{\{ #1 \} }

\newcommand{\Ber}{\mathtt{Ber}}

\renewcommand{\bar}[1]{\widebar{#1}}

\title{Harnessing the Continuous Structure: Utilizing the First-order Approach in Online Contract Design}
\author{
  Shiliang Zuo\thanks{szuo3@illinois.edu}
}
\date{} 

\begin{document}

\maketitle

\begin{abstract}
  
This work studies the online contract design problem. The principal's goal is to learn the optimal contract that maximizes her utility through repeated interactions, without prior knowledge of the agent's type (i.e., the agent's cost and production functions). We leverage the structure provided by continuous action spaces, which allows the application of first-order conditions (FOC) to characterize the agent's behavior. In some cases, we utilize conditions from the first-order approach (FOA) in economics, but in certain settings we are able to apply FOC without additional assumptions, leading to simpler and more principled algorithms. 

We illustrate this approach in three problem settings. Firstly, we study the problem of learning the optimal contract when there can be many outcomes. In contrast to prior works that design highly specialized algorithms, we show that the problem can be directly reduced to Lipschitz bandits. Secondly, we study the problem of learning linear contracts. While the contracting problem involves hidden action (moral hazard) and pricing problem involves hidden value (adverse selection), the two problem share a similar optimization structure, which enables direct reduction between the problem of learning linear contracts and dynamic pricing. Thirdly, we study the problem of learning contracts with many outcomes when agents are identical and provide an algorithm with polynomial sample complexity. 

\end{abstract}

\section{Introduction}

Contract theory has long been a fundamental area of economic research, originating with seminal works such as \cite{holmstrom1979moral}. At its core, contract design concerns how a principal incentivizes an agent to perform tasks on her behalf when the agent’s effort is unobservable. Since the principal cannot directly enforce effort, contracts must be structured to align incentives through performance-based compensation. This problem is central to many real-world applications, such as employee bonuses tied to performance, commission-based compensation in sales and real estate, and incentive structures in digital platforms and gig economies.

Traditional economic research has primarily focused on characterizing optimal contracts under the assumption that the principal has full knowledge of the agent’s production and cost functions (see, e.g., \cite{bolton2004contract}). However, this assumption is often unrealistic. In many practical settings, the principal does not have precise information about how an agent's effort translates into outcomes or the exact cost structure faced by the agent. More recent work \cite{ho2014adaptive} has explored the problem of learning optimal contracts when these functions are unknown, introducing an online learning perspective to contract design.

The repeated interactions between a principal and agents naturally suggest an online learning framework, where the principal iteratively refines contracts based on observed outcomes \cite{ho2014adaptive, Zhu2022TheSC}. While online learning has successfully addressed foundational problems such as Lipschitz bandits and dynamic pricing, contract design is often treated as a separate, more complex domain. Instead of adapting existing techniques, research has largely focused on designing problem-specific algorithms, missing potential structural connections between contract design and well-studied learning frameworks.

We argue that this divergence arises not from inherent complexities in contract design but from a lack of emphasis on the underlying economic structures that could bridge these domains. A key observation is that prior work has largely overlooked the structural advantages provided by continuous action spaces. These spaces introduce natural properties that can simplify the design and analysis of learning algorithms, offering opportunities for more efficient and generalizable solutions. 

This paper investigates the following key question:

\begin{center}
\emph{When we adopt continuous action spaces and the first-order condition (FOC) becomes applicable, what does this reveal about the design of learning algorithms in online learning settings for contract design?}
\end{center}

Our approach is motivated by the observation that adopting continuous action spaces naturally leads to the applicability of the first-order condition (FOC). In economics, the first-order approach (FOA) has been extensively studied as a method for characterizing optimal contracts. This approach relaxes the agent's maximization problem and only requires the agent's response to be a stationary point, rather than the global maximum. The foundational work of Mirrlees (\cite{mirrlees1999theory}, circulated earlier in 1975) demonstrated that the FOA is not always valid, which has led to a substantial body of research examining conditions validating the FOA.

However, from an online learning perspective, the implications of continuous action spaces and the FOC are less explored. In this work, we examine how the presence of the FOC—whether it holds under the conditions developed by the line of work in economics or emerges naturally in certain continuous settings—affects algorithm design for contract learning. We explore whether the structural properties introduced by continuous action spaces enable reductions to well-studied learning paradigms, or can lead to simpler or more principled algorithms.

By identifying and leveraging these structural properties, we aim to bridge the gap between economic theory and online learning, providing a unified perspective on contract design in dynamic environments. Our findings suggest that when the FOC applies, many online contract learning problems exhibit inherent structures that can be systematically exploited, leading to more principled algorithms with improved efficiency and theoretical guarantees.

\subsection{Summary of Results}

\paragraph{Learning Contracts with Many Outcomes for Hetergeneous Agents}

In \Cref{sec:manyOutcomes}, we examine learning algorithms for the principal’s problem of identifying an optimal contract scheme in scenarios with many possible outcomes and heterogeneous agent types. This setting, explored in \cite{ho2014adaptive} and \cite{Zhu2022TheSC}, frames the task as a continuum-armed bandit problem, where the posted contract serves as the \textit{arm}, and the principal observes the realized reward of the selected \textit{arm} in each round. \cite{ho2014adaptive} modified the zooming algorithm for Lipschitz bandits, while \cite{Zhu2022TheSC} introduced a discretization method based on spherical codes.

Despite the structural similarity between online contract design and continuum-armed Lipschitz bandits, these previous works treat the online contract design problem as distinct and propose specialized algorithms. By contrast, we demonstrate that employing the first-order approach allows the online contract design problem to be directly reduced to Lipschitz bandits, thereby enabling a more closer connection between online learning and contract design. 

\paragraph{Learning Contracts with Binary Outcomes (Linear Contracts)}

In \Cref{sec:linear}, we establish a direct equivalence between linear contracts (“contracting”) and posted-price auctions (“pricing”). The contracting problem involves designing optimal linear contracts for agents with hidden actions, while the pricing problem focuses on setting fixed prices for buyers with hidden private valuations \cite{kleinberg2003value}. Although different from an incentive standpoint: contracting addresses moral hazard and pricing deals with adverse selection, we show that both problems share a fundamentally similar optimization structure.

This equivalence allows reductions between the two problems, offering a unified approach. Unlike recent studies such as \cite{dutting2023optimal}, which treat learning optimal linear contracts as a specialized one-sided Lipschitz problem, we leverage first-order conditions and align the problem of contracting with dynamic pricing, enabling more general and efficient solutions.

\paragraph{Learning Contracts with Many Outcomes for Identical Agents}

In \Cref{sec:identical}, we analyze the sample complexity of learning an approximately optimal contract for agents with identical types. This problem has been explored in recent work \cite{chen2024bounded}, which shares some conceptual similarities with our approach. Both methodologies adopt a two-stage framework: in the first stage, the principal learns the agent's private production and cost functions; in the second stage, the principal uses a family of linear programs to identify an approximately optimal contract. However, our algorithm is significantly simpler, relies on more natural assumptions, and achieves superior sample complexity. Our algorithm takes advantage of first-order conditions as well as Grossman and Hart's approach (\cite{grossman1992analysis}), and we believe it to be a more systematic and principled sampling scheme than that of \cite{chen2024bounded}. \footnote{Additionally, \cite{chen2024bounded} assumes CDFP, but for risk-neutral agents, by prior works (e.g. \cite{oyer2000theory}) it is known that the optimal contract often simplifies to a threshold-based bonus contract under CDFP. It is not clear to the author the necessity of their approach given their assumptions, since their assumption already implies some characterization on the optimal shape of the contract. }

\subsection{Organization of paper}
The rest of the paper is organized as follows.  \Cref{sec:model} describes the model and problem formulation. \Cref{sec:manyOutcomes} studies the problem of learning general contracts (when agents can have possibly hetergeneous types). \Cref{sec:linear} specializes the approach in to linear contracts, and yield a connection to the problem of dynamic pricing. \Cref{sec:identical} studies the problem of learning general contracts when agents are identical. \Cref{sec:proofLipschitz} and \Cref{sec:proofsketchIdentical} contains proof or proof sketches of two main results. Finally, \Cref{sec:conclusion} concludes.


\subsection{Additional Related Works}
The problem of contract design is a fundamental topic in economics, dating back to the seminal work of \cite{holmstrom1979moral}, which introduced key concepts related to moral hazard and incentive compatibility. Some subsequent influential studies include \cite{holmstrom1982moral}, which examined moral hazard in team production settings, and \cite{holmstrom1991multitask}, which explored moral hazard in multitask environments. 

The first-order approach has been a widely used method for characterizing optimal contracts, particularly due to its simplicity in handling incentive compatibility constraints. However, \cite{mirrlees1999theory} (circulated in 1975) highlighted critical limitations of this approach, sparking a line of research aimed at identifying conditions under which the first-order approach remains valid. Key contributions in this area include \cite{rogerson1985first} and \cite{jewitt1988justifying}, which established foundational conditions for its applicability, as well as more recent works such as \cite{jung2015information} and \cite{conlon2009two}. Additionally, \cite{grossman1992analysis} proposed an alternative framework that addresses some of the first-order approach's shortcomings.

In recent years, there has been growing interest in studying contract design from computational and learning perspectives within the computer science community. Research by \cite{dutting2019simple, dutting2022combinatorial, dutting2023multi, dutting2023optimal} focuses on combinatorial algorithms for contract design, exploring the complexity and structure of optimal contracts in discrete settings. Complementing this, works such as \cite{ho2014adaptive, Zhu2022TheSC, guruganesh2024contracting, collina2024repeated} investigate learning problems in contract design, applying online learning and algorithmic techniques to dynamic principal-agent models. 

\section{The Model and Preliminaries}
\label{sec:model}




In this section, we first introduce the basic framework of the principal-agent model, and then introduce this problem in the online learning setting. 

\subsection{Basic Components in Principal-Agent Framework}

We consider a principal-agent setting where both parties are risk-neutral. The principal offers a contract to incentivize the agent to exert effort. The agent, in response, makes decisions based on the contract structure and their private cost functions. The problem instance is characterized by the following parameters:

\begin{itemize}

\item $\cA = [0, A_{\max}]$ represents the set of actions available to the agent, corresponding to the level of effort exerted in the underlying task. Higher values indicate greater effort. The agent's effort level directly influences the probability distribution over potential outcomes, governed by the production technology, which in turn determines the principal’s utility.

\item $c: \cA \rightarrow \bbR$ denotes the agent’s cost function, where $c(a)$ specifies the cost associated with exerting effort $a$. The function $c$ is increasing and convex, with $c(0) = 0$, reflecting the economic intuition that exerting greater effort is more costly and that the marginal cost may increase with higher effort levels.

\item $\cY = \set{\pi_1, \dots, \pi_m}$ defines the production outcome space, consisting of $m$ possible discrete outcomes. The $j$-th outcome level is denoted by $\pi_j$, with $0 = \pi_1 < \pi_2 < \dots < \pi_m \leq 1$. Each $\pi_j$ represents the utility received by the principal when the corresponding outcome is realized. 

\item $f: \cA \rightarrow \Delta(\cY)$ describes the agent’s production technology, mapping effort choices to probability distributions over outcomes. Specifically, $f_j(a)$ gives the probability of achieving outcome level $j$ when the agent exerts effort $a$. This function includes the uncertainty inherent in the production process, where higher effort generally increases the likelihood of favorable outcomes but does not guarantee them.

\end{itemize}

\paragraph{Contracts}

A contract defines a mapping from outcomes to payments. We represent the contract as $\vecs: \cY \rightarrow \bbR^+$, where $\vecs$ specifies the payment assigned to each outcome. Importantly, the agent is protected by limited liability, ensuring that payments remain non-negative.

In most cases, we focus on contracts that are both monotone and bounded, as formalized below.

\begin{definition} [Monotone Contracts]

A contract $s$ is monotone if $s_{j+1} \ge s_j$ for any $1 \le j \le m-1$.

\end{definition}

Monotonicity is a natural restriction, as it prevents the agent from sabotaging the outcome (or prevents the percussions when the principal can participate in costless borrowing). This has already been noted in the work of \cite{innes1990limited}. 

\begin{definition} [Bounded Contracts]

A contract $s$ is bounded if $s_j \le 1$ for all $j \in [m]$.

\end{definition}

While the optimal contract is not necessarily bounded, placing an upper limit on payments is often practical. Bounded contracts reflect real-world constraints, such as budget limitations or regulatory requirements. For an in-depth discussion on optimal contracts under bounded payments, see \cite{jewitt2008moral}. Throughout this work, unless stated otherwise, the term \textit{optimal contract} refers to a contract that is both monotone and bounded.

A special class of contracts that we frequently consider are linear contracts, which offer a simple and widely used compensation structure.

\begin{definition}[Linear Contracts]

A linear contract is parameterized by a single value $\beta \geq 0$. Under a linear contract with parameter $\beta$, the principal transfers a fixed proportion $\beta$ of the total utility to the agent. That is, when the $j$-th outcome occurs and the principal's benefit is $\pi_j$, the agent receives a payment of $\beta \pi_j$.

\end{definition}



\paragraph{Utilities} We will write $u_\A(a; \vecs)$ and $u_\P(a; \vecs)$ to represent respectively the expected utility of the agent and principal when the contract is specified by $\vecs$ and the agent takes action $a$. The expected utility of the agent and principal in this model can then be expressed as: 
\[
u_\Ag(a; s) = \sum_{j\in [m]} s_j f_j(a) - c(a), 
\]
\[
u_\pc (a; s) = \sum_{j\in [m]} (\pi_j - s_j) f_j(a). 
\]

For simplicity, we assume the agent's outside option is 0, so that it will not hurt for the agent to participate as long as he is protected by limited liability. In a single round of interaction, the principal posts a contract $s$, and the agent takes an action $a$ satisfying the incentive compatibility constraint 
\[
u_\A(a;s) = \max_{a'\in \cA} u_\A(a';s),
\]
i.e., action $a$ gives the agent the highest expected utility among all actions. For a posted contract $s$, define $u_\P(s)$ as the expected utility to the principal assuming the agent best responds with an incentive compatible action $a$:
\[
u_\P(s) = \max\set{u_\P(a;s) : u_\A(a;s) \ge u_\A(a'; s) }. 
\]


\subsection{Repeated Interactions and Distributional Assumptions on Agent Arrival}
We introduce the sequential interaction setting, where the principal engages with agents over a total of $T$ rounds. In each round $t$, a new agent arrives with a private type. Each agent is characterized by a private cost function $c(\cdot)$ and a production function $f(\cdot)$. We abstract the agent's private type as $\theta_t = (c_t(\cdot), f_t(\cdot))$.

At the start of each round, the principal offers a contract $s_t$. The agent selects a best response based on their private type, and the principal subsequently observes the realized outcome in $\cY$. The principal's objective is to minimize regret, defined as the difference between the cumulative utility achievable under the best fixed contract and the cumulative utility obtained with the contracts posted over $T$ rounds:
\[
\sup_{s^*} \sum_{t=1}^T [ u_\P(s^\ast | \theta_t) - u_\P(s_t | \theta_t) ]. 
\]
Here, $u_\P(s|\theta)$ denotes the utility of posting contract $s$ when the agent has private type $\theta$. 

There are two primary settings of interest: identical agent arrival and heterogeneous agent arrival. In the identical arrival setting, the agent type remains constant across all rounds. The heterogeneous agent arrival setting can be further classified into stochastic and adversarial agent arrival. In the stochastic setting, agent types are drawn from an unknown distribution. In contrast, the adversarial setting allows an oblivious adversary to select agent types without adhering to any fixed distributional assumptions across rounds. 




\subsection{Abel's Lemma}
It will be sometimes more convenient to work with the complementary distribution function, which is the probability of obtaining an outcome no less than level $j$ when the agent takes action $a$. 
\[
G_j(a) := \sum_{k=j}^m f_j(a). 
\]
For any vector $s \in \bbR^m$, denote $\Deltas_j = s_j - s_{j-1}$, where we conveniently define $s_0 = 0$. The following identity is due to Abel's lemma, which will be used in several places throughout the analysis. 
\[
\sum_{j\in[m]} s_j f_j(a) = \sum_{j\in[m]} \Deltas_j G_j(a). 
\]

\szdelete{\color{red}
\subsection{Assumptions}
We state our key assumptions on the problem instance below. 
\begin{assumption}
$c$ is twice differentiable and strongly convex with parameter $\lambda_0$. 
\end{assumption}

The second assumption is that the production technology is Lipschitz continuous with respect to the effort level $a$. This is very natural, it implies the production function $f$ cannot change drastically when the agent shifts his effort by a small amount. 
\begin{assumption}
For any $y\in \cY$, $G_j(a)$ is $L$-Lipschitz in $a$ for any $j$. 
\end{assumption}

The third assumption is that the production outcome satisfies diminishing returns with respect to the effort level $a$. Let $r(a)$ be the expected production output of the agent: 
\[
r(a) = \sum_{j\in [m]} \pi_j f_j(a). 
\]
\begin{assumption}
r(a) is increasing, concave, and twice differentiable in $a$. 
\end{assumption}

\szcomment{remarks TODO}
}

\section{Learning Optimal Contracts with Many Outcomes}
\label{sec:manyOutcomes}

In this section, we examine the problem of learning an optimal contract when multiple outcomes are possible, and the agent’s type may vary across rounds. Prior work has modeled this as a continuum-armed bandit problem \cite{Zhu2022TheSC, ho2014adaptive}, often requiring specialized algorithms that are computationally intensive and challenging to generalize.

We propose an alternative approach using the first-order approach from contract theory, which allows us to reduce the problem to Lipschitz bandits. This not only simplifies the learning process but also ensures strong theoretical guarantees without relying on ad-hoc algorithmic modifications. By leveraging fundamental economic structures, our method directly connects online contract design with well-established learning frameworks, improving both efficiency and generalizability. 

In this section, we first demonstrate how conditions from the first-order approach, particularly Rogerson's CDFC condition, can be leveraged to reveal structural properties of the principal's utility function. We then show how this structure facilitates efficient learning, specifically through a reduction to Lipschitz bandits. Finally, we conclude with a discussion and a detailed comparison with prior works.

\subsection{Structural Property: Continuity of Principal's Utility}

A key insight of our approach is that applying the first-order approach, originally developed for characterizing optimal contracts, can reveal structural properties of the principal’s utility function. By analyzing the agent’s response to a given contract, we establish conditions under which the principal’s utility function exhibits continuity, a crucial property that enables efficient learning.

\begin{assumption} 
\label{assumption:cont}
The following assumptions are imposed in this section. 
\begin{enumerate}
\item Rogerson's CDFC holds: $G''_j(a) \le 0$
\item The cost function of the agent is strongly convex: $c''(a) \ge \lambda$. 
\item  $G_j(a)$ is $L$-Lipschitz with respect to $a$:
\(
\abs{G'_j(a)} \le L. 
\)
\end{enumerate}
\end{assumption}


\begin{proposition}
\label{prop:agentConcave}
Under part 1 of \Cref{assumption:cont}, for any monotone contract $s$, $u_\A(a; s)$ is strictly concave with respect to $a$. 
\end{proposition}

The above proposition precisely implies that the agent's best response to a monotone contract is at a stationary point characterized by the the first-order condition:
\[
\sum_{j=1}^m \Delta s_j G'_j(a) = c'(a). 
\]
Then, by carefully analyzing the agent's best response as characterized by the above, it can be shown that the agent's hidden effort is a continuous with respect to the posted contract $s$. Finally, it can be shown that the principal's utility is continuous. The below theorem shows the Lipschitzness of the principal's utility function, and the proof can be found in \Cref{sec:proofLipschitz}. 




\begin{theorem}
\label{thm:lipschitz}
Under \Cref{assumption:cont}, the utility of the principal $u_\P(s)$ is Lipschitz (as a function of the contract $s$), specifically:
\[
\abs{u_\P(s) - u_\P(s')} \le \left( \frac{4L^2}{\lambda} + 2 \right)\cdot \norm{s - s'}_1. 
\]
\end{theorem}
\begin{remark}
In fact, the result still holds under a much weaker condition. In particular, the theorem still holds as long as $\min_{a} c''(a) - \max_{a,j}G''_j(a) > \lambda$. 
\end{remark}

\subsection{Implications in the Online Learning Setting: Reduction to Lipschitz Bandits}
In this section, we analyze the regret implications of learning contracts in a sequential setting where the principal interacts with a series of agents over $T$ rounds. Each agent possesses a private type $\theta$, which characterizes their production technology and cost function. The principal aims to optimize her contract selection strategy under two different agent arrival scenarios: stochastic agent arrivals and adversarial agent arrivals. Let $u_P(s | \theta)$ represent the utility of a contract $s$ given agent type $\theta$, and assume that \Cref{assumption:cont} holds for all agent types.

Recall in the problem of continuum-armed bandits, a learner chooses an ``arm" in some space (often a continous high-dimensional space), and observes a realized ``reward". The contract design problem then bear close resemblance to the problem, with the contract being treated as the ``arm" and the realized utility as the realized ``reward". Moreover, a prominent line of study is Lipschitz bandits \cite{Kleinberg2013BanditsAE}. The structural result in the previous section precisely enables a direct reduction to Lipschitz bandits. 
\szcomment{rewrite. }

\paragraph{Stochastic agent arrival}
In the stochastic setting, the agent's type in each round is drawn from some fixed distribution $\cD$ that is unknown to the principal. Then, fixing the type $\theta$, the principal's utility $u_\P(s|\theta)$ is Lipschitz. Consequently $\Ex_{\theta\sim \cD}[u_\P(s|\theta)]$ is also Lipschitz. The performance of the principal is measured by the regret, defined as:
\[
T\cdot \max_{s} \Ex_{\theta\sim \cD}[u_\P(s|\theta)] - \sum_{t\in[T]} \Ex_{\theta_t \sim \cD}[u_\P(s_t| \theta_t)]. 
\]
The principal can then directly use the machinery for stochastic Lipschitz bandits (e.g. \cite{Kleinberg2013BanditsAE}) to optimize her regret. 
\begin{proposition}
In the stochastic arrival scenario, the principal can achieve a regret $\widetilde{O}(T^{(z+1)/(z+2)})$. Here $z$ is the zooming dimension of the problem instance (which depends on $\Ex_{\theta\sim \cD}[u_\P(s|\theta)]$). 
\end{proposition}

\paragraph{Adversarial agent arrival}
In the adversarial arrival setting, the agent's type in each round is chosen by some adversary. The regret of the principal is then defined as:
\[
\max_{s} \sum_{t\in[T]} u_\P(s|\theta_t) - \sum_{t\in[T]} u_\P(s_t| \theta_t). 
\]
The principal can directly use the machinery for adversarial Lipschitz bandits (e.g. \cite{Podimata2020AdaptiveDF}) to optimize her regret. 

\begin{proposition}
In the adversarial arrival scenario, the principal can achieve a regret $\widetilde{O}(T^{(z+1) / (z+2)})$. Here $z$ is the adversarial zooming dimension of the problem instance (which depends on the sequence of functions $u_\P(s|\theta_1), \dots, u_\P(s|\theta_T)$). 
\end{proposition}

\subsection{Discussion of Result}
In our results, we assume a continuous action space, which enables the application of ideas from the first-order approach. Originally developed for characterizing optimal contracts in the economics community, we demonstrate that the first-order approach can also provide structural insights that inform the design of learning algorithms. This perspective reveals new connections between economic theory and online learning, offering a more principled framework for algorithm design.

We now compare our approach with two of the most related works on the problem studied in this section. The work by \cite{ho2014adaptive} tackled this problem by proposing a modified version of the zooming algorithm for Lipschitz bandits, relying primarily on the first-order stochastic dominance (FOSD) assumption. A follow-up study by \cite{Zhu2022TheSC} sought to relax these assumptions, presenting an algorithm with fewer structural requirements. However, this relaxation comes at a cost, as the resulting algorithm exhibits a worse regret guarantee, lacks computational efficiency, and relies on a discretization approach based on spherical codes, which is less broadly applicable and may lack generalizability compared to more well-known techniques.

Our approach highlights the importance of leveraging appropriate economic structures. By incorporating insights from the first-order approach, we bridge the gap between economic theory and learning algorithms, demonstrating that this framework allows the problem to be reduced to Lipschitz bandits in a more natural and effective way. This reduction not only simplifies the analysis but also leads to stronger theoretical guarantees, and shows that that by integrating economic intuition into algorithmic design can lead to significant improvements.

\section{Learning Optimal Contracts with Binary Outcomes (Linear Contracts)}
\label{sec:linear}


In this section, we examine the problem of learning linear contracts in a setting where outcomes are binary. 
We establish a structural equivalence between learning linear contracts and posted-price pricing, demonstrating that the two problems can be transformed into one another. While these problems differ in their \emph{incentive} structure—contracting addresses moral hazard due to hidden actions, whereas pricing involves adverse selection due to hidden types—their underlying \emph{optimization} structure is remarkably similar. 


\subsection{Connection between contracting and pricing problems}

This section explores the relationship between contracting and pricing problems by illustrating their structural similarities. We first introduce the two models in more detail, then discuss their equivalence. 

\paragraph{Contracting}
In the contracting setting, a principal offers a linear contract, represented by $\beta \in [0,1]$, to an agent. The outcome of the task is binary: either success ($\pi_2 = 1$) or failure ($\pi_1 = 0$). The agent chooses an effort level $a \in [0,1]$, which determines the probability of success. Higher effort levels increase the likelihood of success but is also associated with higher costs, which recall we denoted as $c(a)$.

When faced with a contract $\beta$, the agent selects an effort level $a$ to maximize their utility, given by $\beta a - c(a)$. The agent’s optimal effort choice satisfies the condition $\beta = c'(a)$, which defines the agent’s best response function $A(\beta)$. The principal, in turn, aims to maximize her own expected utility, given by 
\[
(1 - \beta) A(\beta). 
\]

\paragraph{Pricing}
In the pricing problem, a seller sets a price $p$ for a good, while potential buyers have valuations drawn from a distribution $\mathcal{D}$. The probability that a buyer makes a purchase at price $p$ is given by the demand function $D(p) = \Pr_{v \sim \mathcal{D}}[v \geq p]$. The seller’s expected revenue is then 
\[
p \cdot D(p). 
\]

\paragraph{Structural Equivalence Between the Two Problems}
Despite differences in economic interpretation, the contracting and pricing problems share a common mathematical form. The principal’s utility function in contracting, $(1 - \beta) A(\beta)$, has the same structure as the seller’s revenue function in pricing, $p \cdot D(p)$. By defining $\bar{\beta} = 1 - \beta$ as the share of the outcome retained by the principal, we can rewrite the contracting utility function as 
\[
\bar{\beta} \cdot A(\bar{\beta}). 
\]
Now, this directly mirrors the pricing revenue function 
\[
p \cdot D(p). 
\]
Hence, the two problems share the exact structure from an optimization perspective. 

\begin{multicols}{2}
\begin{framed}
\textbf{Contracting}
    \begin{enumerate} [leftmargin = *]
    \item Agent has cost function $c(\cdot)$
    \item Principal posts contract $\beta$
    \item Agent responds with $a = \arg\max \beta a - c(a)$ (\textbf{hidden action})
    \item Outcome realized as $y \sim \Ber(a)$
    \item Principal's expected utility is $(1-\beta) \cdot A(\beta)$
    \end{enumerate}
\end{framed}
\begin{framed}
\textbf{Pricing}
    \begin{enumerate}[leftmargin = *]
    \item Buyer's valuation distribution is $\mathcal{D}$
    \item Seller posts price $p$
    \item Buyer's valuation $v \sim \mathcal{D}$ (\textbf{hidden type})
    \item Outcome realized as $y = \ind{v \ge p} \sim \Ber(D(p))$
    \item Seller's expected utility is $p\cdot D(p)$
    \end{enumerate}
\end{framed}
\end{multicols}


\subsection{Implications in the Regret-Minimization Online Learning Setting}
This structural equivalence indicates that techniques and insights from the study of dynamic pricing can be leveraged to analyze the problem of learning linear contracts. In particular, learning algorithms developed for dynamic pricing—such as those used for optimizing prices under uncertain demand—can be adapted for online contract design, facilitating more efficient solutions to principal-agent problems.

This section examines the implications of the structural equivalence between contracting and pricing problems in an online learning context. Specifically, we focus on the regret-minimization setting, where the principal aims to optimize contract selection over multiple rounds without prior knowledge of the agent's response function. We first analyze the case where agents are identical, meaning they share the same cost function across all rounds, which parallels the pricing problem where buyers’ valuation distributions remain consistent.

To formalize our results, we introduce the following notation. Let $\mathcal{F}$ be a class of functions. For our application, this will be a class satisfying certain properties (e.g., linear functions, $k$-th order smooth functions, etc. ). Each function in $\mathcal{F}$ represents either the response function in the contracting problem or the demand function in the pricing problem.

Recall we use $A(\beta)$ as the agent's response function in the contracting problem. The regret of an algorithm for the contracting problem is defined as:
\begin{equation} \Reg(T) = \sup_{A\in \mathcal{F}} \sum_{t=1}^T u(\beta^*; A) - u(\beta_t; A), \end{equation}
where $\beta_t$ is the sequence of posted contracts selected by the algorithm, and $\beta^*$ is the optimal contract if the agent's response function $A$ were known. 

The analogous regret definition for the pricing problem is as follows:
\begin{equation}
\Reg(T) = \sup_{D\in \mathcal{F}} \sum_{t=1}^T u(p^*; D) - u(p_t; D).
\end{equation}
where $p_t$ is the sequence of prices selected by the algorithm for the pricing problem, and $p^*$ is the optimal price if the demand function $D$ were known. 

\begin{proposition}
Let $\mathcal{F}$ represent a class of functions, which is known to the learner. The following statements are equivalent:
\begin{itemize}
	\item There exists an algorithm that achieves regret $\Reg(T)$ in the pricing problem. 
	\item There exists an algorithm that achieves regret $\Reg(T)$ in the contracting problem.
\end{itemize}
\end{proposition}

There is a rich body of work on dynamic pricing that establishes regret bounds based on the properties of the demand function $D$, which can belong to different function classes $\mathcal{F}$. For each function class, algorithms are designed to exploit specific structural properties, achieving distinct regret rates. Commonly studied function classes include those with $k$-th order smoothness, infinitely differentiable functions with bounded derivatives, and others (e.g., \cite{wang2021multimodal}). Through our reduction, these results directly translate to regret bounds in the contracting problem, as the agent's response function (or cost function) can similarly belong to various function classes. This broad applicability of our reduction between pricing and contracting provides a unified framework to derive optimal regret bounds for contracting across different function classes. 


As a concrete example, consider the foundational results in \cite{kleinberg2003value}, where they show that if the pricing problem has a unique maxima and that the seller's revenue is locally strongly concave around the unique maxima, then the learner can achieve a $\widetilde{O}(\sqrt{T})$ regret, and that further no algorithm can achieve $o(T)$ regret. This directly translates to regret bounds in the contracting problem as follows. 

\begin{proposition}
Consider the contracting problem. If the principal's utility function $u(\beta)$ has a unique global maximum at $\beta^*$ and that $u''(\beta) > 0$, then there exists an algorithm that achieves regret $\widetilde{O}(\sqrt{T})$. Further, no algorithm can achieve regret $o(\sqrt{T})$. 
\end{proposition}

\paragraph{Extending to Stochastic Agent Arrival}
We show that the previous argument for identical agent arrival can extend directly to stochastic agents. In the stochastic arrival setting, the cost function $c(\cdot)$ of the agent each round is drawn from some unknown distribution. For each possible agent's cost function $c(\cdot)$, denote $A(\beta; c(\cdot))$ as the agent's response function. Finally, let us denote 
\[
\bar{A}(\beta) = \Ex_{c}[A(\beta; c(\cdot))]
\]
Then $\bar{A}(\beta)$ denotes the agent's expected response function, and the results extend from the identical case. 


\paragraph{Extending to Adversarial Agent Arrival}
In the adversarial arrival setting, the agent's type each round need not satisfy any distributional assumptions. This becomes equivalent to a pricing problem where the agent's value distributions are chosen by an oblivious adversary each round. We can again apply the result from \cite{kleinberg2003value} and achieve a $O(T^{2/3})$ regret for the contracting problem in the adversarial arrival setting. 

\subsection{Discussion of Results}

Although prior work has studied linear contracts from a computational perspective, most assume a discrete action space for the agent. For example, \cite{dutting2019simple} examined the implementability of actions under linear contracts. As a result, they introduce cumbersome notations and complex arguments to determine which action is implemented under a given linear contract. Specifically, their analysis involves terms like $(c(a') - c(a)) / (a' - a)$ (using our notation), which just simplifies to the gradient when $a$ comes from a continuous space. However, since they shy away from working with continuous action spaces, they complicate what should be a straightforward observation: the agent’s action is at a stationary point determined by the first-order condition. By explicitly leveraging this characterization—which, while straightforward, uncovers meaningful structural insights—we establish an optimization-based equivalence between contracting and pricing. While mathematically simple, to the best of our knowledge, this result has not been explicitly stated in prior work. 

Consider the prior work \cite{dutting2023optimal} which studied the problem of learning one-sided Lipschitz functions, motivated by learning linear contracts when the agent’s action space is finite. While their algorithm is interesting, it appears highly specialized. Further, their algorithm assumes deterministic observations. By contrast, our approach shows that the problem of learning linear contracts is \emph{equivalent} to the well-studied dynamic pricing problem, allowing the two problems to be reduced to one another. This connection is simple yet fundamental, and it appears that prior works have overlooked this key relationship.

We also briefly compare our results with \cite{Zhu2022TheSC}, which studied the problem of learning linear contracts and provided a lower bound of $\Omega(T^{2/3})$. There is no contradiction, as the difference arises from distinct modeling choices. Specifically, their lower bound is derived from a construction where the agent has an exponential number of discrete actions with artificially designed cost values. In contrast, we assume a continuous action space with a convex cost function, leading to a fundamentally different structure.

\section{Learning Contracts for Many Outcomes with Identical Agents}
\label{sec:identical}

\begin{algorithm}
\caption{Polynomial Sample Complexity for Identical Agents}
\label{algo:polynomialForIdentical}
\begin{algorithmic}
\State \szcomment{Change this to directly estimate CDF}
\State Input: $\eps, \delta$. Algorithm returns $\eps$-approximately optimal contract with probability $1-\delta$. 
\State Set $\eps_c = \frac{\lambda \eps}{2L^2}, N = \frac{2\log( \beta_{\max}/(\delta \eps_c) )}{\eps^2}$
\State Set $\C_D = \set{0,1,\dots, \floor{\beta_{\max}/\eps_c}} \cdot \eps_c$, the discretized contract space
\State // \emph{First Stage}: Learn estimates of the agent's production and cost function by querying linear contracts
\State For each contract $\beta^\i \in \C_D$, the principal queries this contract $N$ times
\State Principal records the empirical complementary CDF $\hat{G}_j(a^\i)$\\
\State Principal computes the empirical average outcome using $\hat{G}$:
\[
\hat{r}(a^\i) = \Ex_{j\sim \hat{G}(a^\i)}[\pi_j]
\]

\State Principal compute the confidence interval for the cost as:
\begin{align*}
{c}^\LCB(a^\i) &:=  \sum_{j=1}^i \beta^{(j-1)} (\hat{r}(a^\j) - \hat{r} (a^{(j-1)})) - 3\eps, \\
{c}^\UCB(a^\i) &:= \sum_{j=1}^i \beta^\j (\hat{r}(a^\j) - \hat{r}(a^{(j-1)})) + 3\eps
\end{align*}
\State // \emph{Second Stage}: Solve LPs using the estimated parameters
\State Compute
\begin{align*}
    u^\UCB_\A(a^\i; s) := \sum_{j\in[m]}\Deltas_j G^\UCB_j(a^\i) - c^\LCB(a^\i) \\
    u^\LCB_\A(a^\i; s) = \sum_{j\in[m]}\Deltas_j G^\LCB_j(a^\i) - c^\UCB(a^\i). 
\end{align*}
\State For each $a^\i \in \cA_D$, define the following program $\LP(a^\i)$, denote its value by $\APX(a^\i)$:
\begin{align*}
\max_s \quad& \hatr(a^\i) - \sum_{j\in[m]} s_j \hatf_j(a^\i) \\
u_\A^\UCB(a^\i; s) &\ge u_\A^\LCB(a^\k; s) - \eps \quad \forall a^\k \in \cA_D \\
s_j & \ge s_{j-1}, s_j \ge 0
\end{align*}
\State Return the solution to the linear program with maximum $\APX(a^\i)$
\end{algorithmic} 
\end{algorithm}

In this section, we study the problem of learning the optimal contract when agents are identical, meaning each agent has the same cost and production function. The number of possible outcomes may be finite (not necessarily binary) or even infinite. Our approach integrates ideas from inverse game theory, the first-order approach, and Grossman and Hart's framework for contract design. The results in this section will be more suitable stated in terms of sample complexity rather than regret. 

At a high level, the principal first conducts inverse learning to infer the agent's production and cost functions. This involves selecting contracts that induce specific agent actions, allowing the principal to observe the associated costs and production outcomes. By systematically varying the contract terms, the principal effectively recovers the agent's underlying production and cost structure.

For general and complex contracts, predicting the agent's best response can be challenging. However, within a certain class—specifically, linear contracts—the agent’s best response can be precisely characterized using the first-order condition (as we also noted in the previous section). By leveraging linear contracts, the principal can induce a range of agent actions and infer the agent’s private type. Once this inverse learning phase is complete, the principal applies Grossman and Hart's approach, solving a family of linear programs to identify an approximately optimal contract.

\begin{assumption} 
The following assumptions are made in this section. 
\begin{enumerate}
\item $c(a)$ is twice-differentiable and $\lambda$-strongly convex. 
\item $G_j(a)$ is $L$-Lipschitz. 
\item $r(a)$ is strictly increasing, concave, and twice differentiable in $a$. 
\item The maximum level of effort $A_{\max} \le \frac{L}{\lambda}$ and can be implemented by the linear contract $\beta_{\max}$. 
\end{enumerate}
\end{assumption}
\begin{remark}
The third part of the assumption amounts to the diminishing returns intuition. The fourth part is without loss of generality, since the cost of implementing any effort level greater than $\frac{L}{\lambda}$ already outweighs the expected outcome, hence cannot be optimal to the principal. In studying the sample complexity, we will mostly treat $\lambda, L, \beta_{\max}$ as constants and focus on the dependence on $\eps$. 
\end{remark}

\subsection{Algorithm and Main Result}

The algorithm consists of two stages: learning the agent's cost and production functions, followed by solving a family of linear programs to determine an approximately optimal contract.

In the first stage, the principal discretizes the linear contract space $[0, \beta_{\max}]$ and queries each contract in the discretized space sufficiently many times. This process allows the principal to estimate the agent's production and cost functions with a small error. To be more specific, the principal can get an estimated production technology and cost value for a discretized action space. Note that the exact ``identitity" of the value of $a$ is unimportant, and that only the associated cost value and production function will be of relevance in the next stage. 

In the second stage, the principal uses these estimates to solve a family of linear programs (\cite{grossman1992analysis}), enabling the identification of an approximately optimal contract. In particular, for each action in the discretized action space, the principal can use a linear program to find the minimum payment contract for which the action $a$ is approximately incentive-compatible using the learnt production technologies and costs on the discretized space. Then, ranging over all such actions in the discretized space give a family of linear programs, and the principal can find an approximately optimal and approximately incentive-compatible contract by solving every linear program. 

The final step involves converting an approximately incentive-compatible contract into a exact incentive-compatible contract by utilizing an interesting observation in \cite{dutting2021complexity}. 

One potential drawback of this approach is that $\beta_{\max}$ may exceed 1, meaning the principal might need to offer contracts where the agent's payment exceeds the actual outcome. This issue arises only in the learning phase and does not affect the final contract selection.

\begin{theorem}
\label{thm:sample}
The principal can find a $\eps$-optimal contract using $\widetilde{O}(1 / \eps^6)$ samples. 
\end{theorem}

\begin{remark}
Note that the sample complexity is independent of the number of outcomes $m$. As a consequence, our scheme can be adapted to the case where the outcome space is continuous or infinite, e.g. $\cY = [0,Y]$, a continuous interval. Specifically, this involves discretizing the outcome space and applying our algorithm to the discretized outcome space. Bounding the discretization error in the outcome space requires some additional regularity conditions on the production function. In particular, if we write $G_y(a)$ to be the production function, then $G$ should not only be continuous in $a$T, but also continuous in $y$. The exact details are omitted. 
\end{remark}

\subsection{Discussion of Result}
Our approach combines the best of both the first-order approach and Grossman and Hart’s approach. While the first-order approach provides a way to characterize the agent’s best response, it often requires strong assumptions. However, we observe that the agent's best response can be precisely characterized when using linear contracts . This, in turn, enables the application of Grossman and Hart’s approach after the inverse learning phase, which do not require any extra assumptions. \szcomment{rewrite. }

Our algorithm design also reflects some insightful remarks from Holmström's Nobel lecture. Holmström noted that optimal contracts are often complex due to the ``imbalance between the agent’s one-dimensional action space and the principal’s high-dimensional control space (i.e., the contract space)". To address this challenge, the economics community developed the first-order approach, which enables the characterization of optimal contracts to be more tractable. Furthermore, Holmström highlighted that focusing on linear contracts leads to a more well-behaved model because "the one-dimensional effort space of the agent and the one-dimensional control space of the principal are evenly matched." Our algorithm builds on this observation. In the learning phase, we employ linear contracts to effectively control the agent’s actions, allowing for accurate estimation of the agent’s production and cost functions. In the second stage, we leverage the information gathered to compute an optimal contract without imposing additional assumptions. 


Finally, we compare our results with the recent work by \cite{chen2024bounded}. Our approach offers three key advantages. First, our algorithm is simpler and more straightforward, involving significantly fewer complicated steps. Second, we rely on milder, more natural assumptions that are stated intuitively. Specifically, we do not require assumptions that need "extra justification"—our main assumption is the standard diminishing returns condition. Third, our algorithm achieves superior sample complexity that does not depend on the number of outcomes and has a better dependence on $\varepsilon$. Furthermore, their work operate under a type of Rogerson's CDFC as in the his work \cite{rogerson1985first} (although they did not explicitly point out as such); as well as the FOSD assumption. It is well-known that under CDFC and MLRP, the optimal contract for fixed agent types is a threshold-type bonus contract (e.g. \cite{oyer2000theory}). It is unclear to the author whether their approach is entirely necessary given their assumptions; while FOSD is weaker than MLPR, the assumptions nevertheless seem to already suggests that there is some limit on the optimal contract's particular shape.

\section{Proofs for Continuity Conditions in \Cref{sec:manyOutcomes}}
\label{sec:proofLipschitz}
\begin{proof} [Proof of \Cref{prop:agentConcave}]
Recall the agent's utility can be expressed as:
\begin{align*}
u_\A(a; s) &= \sum_{j=1}^m \Deltas_j G_j(a) - c(a).   
\end{align*}
Taking the second derativie with respect to $a$:
\begin{align*}
u''_\A(a;s) &= \sum_{j=1}^m \Deltas_j G''_j(a) - c''(a) \\
&< \sum_{j\in[m]} \Deltas_j c''(a) - c''(a) \\
&\le s_m c''(a) - c''(a) \\
&\le 0. 
\end{align*}
This completes the proof. 
\end{proof}

\begin{lemma}
The agent's response $a(s)$ is Lipschitz continuous with respect to the contract $s$ in the monotone contract domain, specifically, for any two monotone contracts $s^1, s^2$:
\[
\abs{a(s^1) - a(s^2)} < \frac{2L}{\lambda - \mu} \norm{s^1 - s^2}_{1}. 
\]
\end{lemma}
\begin{proof}

Let $a := a(s^1), b := a(s^2)$ be the best responses. Without loss of generality assume $a < b$. 

By the previous \Cref{prop:agentConcave}, the optimal response must satisfy the first-order conditions. 
\begin{align*}
\sum_{j\in [m]} s^1_j f'_j(a) &\le c'(a) \\
\sum_{j\in[m]} s^2_j f'_j(b) &\ge c'(b). 
\end{align*}
The above expressed in terms of $G$ and $\Deltas$ becomes:
\begin{align*}
\sum_{j\in[m]} \Deltas^1_j G'_j(a) &\le c'(a) \\
\sum_{j\in [m]} \Deltas^2_j G'_j(b) &\ge c'(b). 
\end{align*}
Note that by strong convexity of $c$, the difference between the right-hand side can be lower bounded: 
\[
c'(b) - c'(a) \ge \lambda (b-a). 
\]
The difference between the left-hand side can be upper bounded: 
\begin{align*}
\sum_{j\in[m]} \abs{ \Delta s^1_j G'_j(a) - \Delta s^2_j G'_j(b)  } &= \sum_{j\in[m]} \abs{\Delta s^1_j G'_j(a) - \Delta s^2_j G'_j(a) - \Delta s^2_j (G'_j(b) - G'(a))} \\
&\le \sum_{j\in[m]} \abs{\Delta s^1_j - \Delta s^2_j} \abs{G'_j(a)} + \sum_{j\in[m]} \abs{\Delta s^2_j} \abs{G'_j(b) - G'_j(a)} \\
&\le 2L \norm{s^1 - s^2}_1 + \mu\abs{a-b}.
\end{align*}

Therefore, we must have:
\begin{align*}
2L\norm{s^1 - s^2}_1 + \mu\abs{a-b} \ge \lambda \abs{a-b}
\end{align*}
which translates to:
\[
\abs{a-b} < \frac{2L\norm{s^1 - s^2}_1}{\lambda - \mu}. 
\]
Hence, when $s$ is restricted to be monotone, the response $a(s)$ is Lipschitz with respect to the contract $s$. 
\end{proof}

\begin{theorem}
As a function of $s$, the utility of the principal $u_\P(s)$ is Lipschitz, specifically:
\[
\abs{u_\P(s^1) - u_\P(s^2)} \le \left( \frac{4L^2}{\lambda - \mu} + 2 \right). \norm{s^1 - s^2}_1. 
\]
\end{theorem}
\begin{proof}

Recall the utility of the principal can be expressed as:
\begin{align*}
u_\P(s) &= \sum_j (\pi_j - s_j) f_j(a(s)) \\
&= \sum_{j\in[m]} \pi_j f_j(a(s)) - \sum_{j\in[m]} s_j f_j(a(s))
\end{align*}
Consider the two terms separately. The first term is the expected outcome, and can be expressed as:
\begin{align*}
\sum_j \pi_j f_j(a(s)) &= \sum_{j\in[m]} (\pi_j - \pi_{j-1}) G_j(a(s)) 
\end{align*}
which can been seen to be $\frac{2L^2}{\lambda - \mu}$-Lipschitz. This follows from the fact that $a(s)$ is $\frac{2L}{\lambda - \mu}$-Lipschitz with respect to $s$, and that $G_j(a)$ is $L$ Lipschitz with respect to $a$. 

Next, consider the second term $\sum_{j\in[m]} s_j f_j(a(s))$, which is the expected payment to the agent. Consider two contracts $s^1, s^2$ and let the best response be $a, b$ respectively. 
Then the difference between the payments can be upper bounded: 
\begin{align*}
\abs{ \sum_{j\in [m]} s^1_j f_j(a) - \sum_{j\in [m]} s^2_j f_j(a) } &\le \sum_{j\in [m]} \abs{ \Delta s^1_j G_j(a) - \Delta s^2_j G_j(b) } \\
&= \sum_{j\in [m]} \abs{ \Delta s^1_j G_j(a) - \Delta s^1_j G_j(b) - (\Delta s^2_j-\Delta s^1_j) G_j(b)}  \\
&\le \sum_{j\in [m]} \abs{\Deltas_j^2}\abs{G_j(b) - G_j(a)} + \sum_{j\in [m]} \abs{\Deltas_j^1 - \Deltas_j^2}\abs{G_j(b)} \\
&\le \paren{\frac{2L^2}{\lambda - \mu}}\norm{s^1 - s^2}_1 + 2 \norm{s^1 - s^2}_1 \\
&= \left( \frac{2L^2}{\lambda - \mu} + 2 \right) \norm{s^1 - s^2}_1
\end{align*}
In the fourth line, we use the fact that $\abs{G_j(b) - G_j(a)} \le L\abs{b-a}$ and $\abs{b-a} \le (\frac{2L}{\lambda - \mu})\norm{s^1 - s^2}_1$ by the previous lemma. This shows the second term is $\left(\frac{2L^2}{\lambda - \mu} + 2\right)$-Lipschitz. 

The utility to the principal is the expected outcome minus payment, hence must be $\left(\frac{4L^2}{\lambda - \mu} + 2\right)$-Lipschitz. 
\end{proof}

\section{Proof Sketch for \Cref{thm:sample}}
\label{sec:proofsketchIdentical}
In this section we give a proof sketch for \Cref{thm:sample} and state the key lemmas in the proof. We break the proof into several steps. 
\paragraph {\bf Step 1. } In the first stage of the algorithm, the principal queries linear contracts from discretized points inside some interval. Under a very natural diminishing returns assumption, it can be shown that the agent's response will be continuously increasing as the linear contract increases. 

Let us denote $\beta{(i)} = i\cdot \eps$, and let $a^{(i)}$ to represent the induced action under the linear contract $\beta^{(i)}$. We can show that the induced actions are sufficiently ``dense" in the action space. 


\begin{restatable}{lemma}{lemmaDiscretizedActionClose}
Let $\beta' > \beta$, then $a(\beta') - a(\beta) < \frac{r'(0)}{\lambda} (\beta' - \beta)$. 
\end{restatable}

The production function associated with each contract can be estimated directly by observing the realized production outcome; by applying standard concentration results, the errors in the estimates can be controlled. When each linear contract in the discretized space is queried enough times, the principal can get an accurate estimate for the production technology function at each point. 
\begin{restatable}{lemma}{lemmaEstimateAccurate}
Fix the linear contract $\beta^\i \in \C_D$ and let the induced action be $a^\i$. Suppose the contract has been queried for $N = \frac{\ln(2/\delta)}{\eps^2}$ times. With probability $1 - \delta$,
\[
\abs{ \hatG_j(a^\i) - G_j(a^\i) } \le \eps, ~~ \abs{ r(a^\i) - \hatr(a^\i)} \le \eps. 
\]
\end{restatable}

However, the cost function cannot be estimated directly. Here, an application of Myerson's lemma is applied to obtain an accurate estimate of the cost function. 

\begin{restatable}{lemma}{lemmaMyerson}
Let $a = a(\beta)$ be the action induced by $\beta$, $a' = a(\beta')$ be the action induced by $\beta'$. Then the following holds:
\[
\beta' [r(a) - r(a')] \le c(a) - c(a') \le \beta [r(a) - r(a')]
\]
\end{restatable}


Utilizing the above, one can obtain confidence bounds on the cost function values for the agent. 

\begin{restatable}{lemma}{lemmaActionConfidenceInterval}
Let ${c}^\LCB(a^\i), {c}^\UCB(a^\i)$ be as defined in the algorithm. Then for any $i$
\[
c(a^\i) \in \left[ {c}^\LCB(a^\i),   {c}^\UCB(a^\i) \right]. 
\]
Further 
\[
{c}^\UCB(a^\i) - {c}^\LCB(a^\i) \le 9\eps. 
\]
\end{restatable}

\paragraph {\bf{Step 2. }} In the second stage of the algorithm, the principal solves a family of LPs using the estimated parameters. Since there may be small errors in the estimates, the computed contract may not be exactly incentive compatible. However, we will be able to bound the `incentive gap' by carefully analyzing the deviation terms. Note that we also carefully left the apprriopate adjustments in the linear programs to adjust for the deviation in the confidence intervals. 

Let us consider the true optimal contract $s^*$, let the optimal induced action under $s^*$ by $a^*$. We can show that the optimal contract $s^*$ is at least feasible solution to the action $\bar{a^*}$, which is the closest action to $a^*$ in the discretized space. 
\begin{restatable}{lemma}
{lemmaSStarFeasible}
$s^*$ is a feasible solution to $\LP(\bar{a^*})$. 
\end{restatable}

By carefully analyzing the deviation terms, we can show $a^\i$ will be $\delta$-IC under any solution to the program $\LP(a^\i)$. 
\begin{restatable}{lemma}
{lemmadeltaIC}
\label{lemma:deltaIC}
Consider the program $\LP(a^\i)$ and let $s$ be a solution. Then $a^\i$ is $\delta$-IC under $s$ with $\delta = 22\eps$. 
\end{restatable}{lemma}

Now, let $a^*_D$ be the action with the highest objective among the family of $\LP(a^\i)$, and let $s^*_D$ be the solution to $\LP(a^*_D)$. We can show this solution is very close to the optimal value achieved by the optimal contract. 
\begin{restatable}{lemma}
{lemmaOPTbound}
\label{prop:approxICandOPT}
Under the contract $s_D^*$, the action $a_D^*$ is $\delta$-IC with $\delta = 22\eps$. Moreover, $\OPT \le u_\P(a^*_D; s^*_D) + 6\eps$. 
\end{restatable}

\paragraph {\bf Step 3. } Finally, we apply the interesting observation in \cite{dutting2021complexity}. 
\begin{lemma*}[\cite{dutting2021complexity}]
Let $(s,a)$ be a $\delta$-IC contract. Then $s' := (1-\sqrt{\delta})s + \sqrt{\delta} \pi$ achieves utility at least $(1-\sqrt{\delta}) u_\P(a; s) - (\sqrt{\delta} - \delta)$. 
\end{lemma*}
This allows the principal to convert any approximately-IC contracts to exact IC contracts, and the final result is obtained.

\section{Conclusion and Discussion}
\label{sec:conclusion}


While recently there are many works in computer science studying the economic problem of contract design, a modelling choice has to be made: continuous effort space mode, or the discrete action space model. Most theoretical computer science (TCS) approaches to contract design adopt discrete action space models to introduce combinatorial structures that make the problem more tractable or computationally interesting. A recent survey \cite{dutting2024algorithmic} also suggests that discrete models are more natural for computer scientists. While this may hold for complexity theorists, we argue that for machine learning researchers, the continuous action space model is more natural and practical. While the distinction isn’t strictly binary, continuous models align more closely with real-world learning scenarios and allow for more elegant, generalizable algorithms, as demonstrated in our work. In addition to being the default in economics, the continuous action space model serves as a natural bridge between online learning and contract design. We’ve demonstrated that leveraging continuous action spaces, combined with insights from the FOA, leads to simpler, more principled, and broadly applicable algorithms.

Choosing the continuous action space model, we studied three key problems in online contract design, presenting three main results. First, inspired by the FOA, we identified conditions under which the problem of learning optimal contracts with heterogeneous agents can be reduced to Lipschitz bandits. Second, we showed that the problem of learning the optimal linear contract is equivalent to dynamic pricing with an unknown demand curve. Third, we proposed an algorithm with polynomial sample complexity for learning the optimal contract when dealing with identical agents. All results are based on the continuous effort model, which is the default in economics but has been largely overlooked in recent computer science literature. Our findings highlight that the continuous effort model enables the application of a broader set of techniques in online contract design, particularly those grounded in the first-order approach from economics. This perspective also provides new insights for designing online learning algorithms. 

\newpage
\bibliographystyle{ACM-Reference-Format}  
\bibliography{references}  

\appendix

\section{Proofs for Contracting with Identical Agents in \Cref{sec:identical}}
\subsection{The estimated paramters are accurate}

We first show the response of the agent is well-behaved under linear contracts. 
\begin{restatable}{lemma}{lemmaLinearSpans}
Let $a(\beta)$ denote agent's best response when the linear contract $\beta$ is posted. Then $a(\beta)$ is increasing and continuous in $\beta$. 
\end{restatable}
\begin{proof}
Recall $r(a)$ is the expected production of the agent. Assume the principal posts a linear contract $\beta$, the agent then takes the action $a = a(\beta)$ that maximizes the following:
\[
\max_a \beta r(a) - c(a). 
\]
Note that this is a concave function in $a$. 
The optimal solution then satisfies the following. 
\begin{enumerate}
\item If $\beta < \frac{c'(0)}{r'(0)}$, the agent's best response is $a(\beta) = 0$,
\item If $\beta > \frac{c'(E)}{r'(E)}$, the agent's best response is $a(\beta) = E$, 
\item Otherwise, the agent's best response is the unique actions satisfying:
\[
\frac{c'(a)}{r'(a)} = \beta. 
\]
\end{enumerate}
Note $c'$ is increasing, $r'$ is decreasing and both $c'$ and $r'$ are continuous. Hence, the agent's response is continuous and increasing with respect to $\beta$. The principal is then able to implement every action using linear contracts. 
\end{proof}

When the contract space is appropriately discretized, the induced actions are also close to each other. In other words, $a(\beta)$ is Lipschitz with respect to $\beta$. 

\lemmaDiscretizedActionClose*
\begin{proof}
With out loss of generality let us assume $\frac{c'(0)}{r'(0)} < \beta < \beta' < \frac{c'(E)}{r'(E)}$. Recall at each linear contract $\beta \in [\frac{c'(0)}{r'(0)}, \frac{c'(E)}{r'(E)}]$, the induced action $a$ satisfies:
\[
\frac{c'(a)}{r'(a)} = \beta. 
\]
Consider $\beta$ as a function of $a$, then
\begin{align*}
\beta'(a) &= c'(a)\cdot \frac{-r''(a)}{(r'(a))^2} + c''(a)\cdot\frac{1}{r'(a)} \\
&\ge c''(a)\cdot\frac{1}{r'(a)} \\
&\ge \frac{\lambda}{ r'(0) }. 
\end{align*}
Hence, by the inverse function rule:
\[
a'(\beta) \le \frac{r'(0)}{\lambda}. 
\]
\szcomment{If only assume $r'' < \mu$, then $c'(a)\cdot -\mu/r'^2 + \lambda/L$}
This is enough to complete the proof. 
\end{proof}

By querying a specific contract $\beta^\i$ enough times, the principal can learn an accurate estimate of the production technology $f_j(a^\i)$ (and $G_j(a^\i)$).

\lemmaEstimateAccurate*
\begin{proof}
Follows directly from the Dvoretzky–Kiefer–Wolfowitz (DKW) inequality. 
\end{proof}

There is another question we need to solve: how can the principal infer the cost of the agent's action under some contract $\beta$? This can be solved using an adaptation of Myerson's lemma.

\lemmaMyerson*
\begin{proof}
By incentive compatibility:
\begin{align*}
\beta r(a) - c(a) &\ge \beta r(a') - c(a') \\
\beta' r(a') - c(a') &\ge \beta' r(a) - c(a)
\end{align*}
The lemma then follows by rearranging terms. 
\end{proof}

The below lemma characterizes the confidence interval for the cost of the action $a(\beta^\i)$. 

\lemmaActionConfidenceInterval*
\begin{proof}
For any $i \ge 1$, the following holds. 
\begin{align*}
c(a^\i) &= \sum_{j=1}^i c(a^\j) - c(a^{(j-1)}) \\
&\le \sum_{j=1}^i \beta^\j (r(a^\j) - r(a^{(j-1)})) \\
&= \beta^\i r(a^\i) + \sum_{j=2}^i (\beta^{(j-1)} - \beta^\j) r(a^{(j-1)}) - \beta_1 r(a^{(0)}) \\
&\le \beta^\j \hat{r}(a^\j) + \sum_{j=2}^i (\beta^{(j-1)} - \beta^\j) \hat{r}(a^{(j-1)}) - \beta_1 \hat{r}(a^{(0)})+ 3\eps \\
&\le \sum_{j=1}^i \beta^\j (\hat{r}(a^\j) - \hat{r}(a^{(j-1)})) + 3\eps. 
\end{align*}
Similarly we can obtain a lower bound for $c(a^\i)$. 
\begin{align*}
c(a^\i) &= \sum_{j=1}^i c(a^\j) - c(a^{(j-1)}) \\
&\ge \sum_{j=1}^i \beta^{(j-1)} (r(a^\j) - r(a^{(j-1)})) \\
&= \beta^{(i-1)} r(a^\i) + \sum_{j=2}^i (\beta^{(j-2)} - \beta^{(j-1)}) r(a^{(j-1)}) - \beta^{(0)} r(a^{(0)}) \\
&\ge \beta^{(i-1)} \hat{r}(a^\i) + \sum_{j=2}^i (\beta^{(j-2)} - \beta^{(j-1)}) \hat{r}(a^{(j-1)}) \beta^{(0)} - \beta^{(0)}\hatr(a^{(0)}) - 3\eps\\
&\ge \sum_{j=1}^i \beta^{(j-1)} (\hat{r}(a^\j) - \hat{r} (a^{(j-1)})) - 3\eps.  \\
\end{align*}

Further,
\begin{align*}
    {c}^\UCB(a^\i) - {c}^\LCB(a^\i) &\le 
    \sum_{j=1}^i (\beta^\j - \beta^{(j-1)}) (\hatr(a^\j) - \hatr(a^{(j-1)})) + 6\eps \\
    &\le \sum_{j=1}^i (\beta^\j - \beta^{(j-1)}) 3\eps + 6\eps \\
    &\le 3\eps + 6\eps \\
    &\le 9\eps,
\end{align*}
where in the second line we use the fact that:
\begin{align*}
\abs{\hatr(\beta^\j) - \hatr(\beta^{(j-1)})} &\le 2\eps + r(\beta^\j) - r(\beta^{(j-1)}) \\
&\le 2\eps + r'(0)\cdot (a(\beta^\j) - a(\beta^{(j-1)})) \\
&\le 2\eps + r'(0)\cdot \frac{r'(0)}{\lambda} \cdot \eps_c \\
&\le 3\eps. 
\end{align*}
This finishes the proof. 
\end{proof}

\subsection{The family of linear programs gives an approximately optimal contract}

Recall for each $i\in \set{0, 1, \dots, \down{1/\eps_c}}$, the principal solves the following linear program, denote the program by $\LP(a^\i)$ and its value by $\APX(a^\i)$. 
\begin{align*}
\max_s &\quad \hatr(a^\i) - \sum_{j\in[m]} s_j \hatf_j(a^\i) \\
u^\UCB_\A(a^\i; s) &\ge u^\LCB_\A(a^\k; s) - O(\eps) \quad \forall a^\k \in \cA_D \\
s_j &\ge s_{j-1} \\
s_j &\ge 0
\end{align*}

\begin{lemma}
\label{lemma:discretizeUtilityClose}
Fix any action $a$ and any monotone bounded contract $s$. Let $\bar{a}$ be the action in $\cA_D$ that is closest to $a$. Then, 
\[
u_\A(a;s) - u_\A(\bar{a};s) \le \eps. 
\]
\end{lemma}
\begin{proof}
We have
\begin{align*}
\abs{ u_\A(a;s) - u_\A(\bar{a};s) } &\le \abs { \sum_{j\in [m]} s_j [f_j(a) - f_j(\bar{a})] } + \abs{c(a) - c(\bar{a})} \\
&\le \abs { \sum_{j\in [m]} \Deltas_j [G_j(a) - G_j(\bar{a})] } + \abs{c(a) - c(\bar{a})} \\
&\le \sum_{j\in[m]}\Deltas_j \cdot L\cdot \abs{a - \bar{a}} + c'(E)\cdot \abs{a - \bar{a}} \\
&\le (L + c'(E))\cdot \abs{a - \bar{a}} \\
&\le (L + c'(E))\cdot \frac{r'(0)}{\lambda} \eps_c\\
&\le \frac{2L^2}{\lambda} \eps_c \\
&\le \eps. 
\end{align*}
Here in the last line, we used the fact that:
\[
\eps_c \le \frac{1}{L + c'(E)}\cdot \frac{\lambda}{r'(0)}. 
\]
\szcomment{Need condition here again. }
\end{proof}

\begin{lemma}
\label{lemma:utilityConfidenceBound}
$0 \le u^\UCB(a^\i;s) - u_\A(a^\i;s) \le 10\eps$, $0 \le u_\A(a^\i;s) - u^\LCB_\A(a^\i;s) \le 10\eps$. 
\end{lemma}
\begin{proof}
\begin{align*}
u^\UCB(a^\i;s) - u_\A(a^\i;s) &= \sum_{j\in[m]}\Deltas_j [G^\UCB_j(a^\i) - G_j(a^\i)] + c(a^\i) - c^\LCB(a^\i) \\
&\le 10\eps. 
\end{align*}
And similarly for others. 
\end{proof}

We now show $a^\i$ will be $\delta$-IC under any solution to the program $\LP(a^\i)$. 
\lemmadeltaIC*
\begin{proof}
Fix any action $a$, then by the previous \Cref{lemma:discretizeUtilityClose}, there exists some $\bar{a}\in\cA_D$ such that
\[
u_\A(a;s) \le u_\A(\bar{a};s) + \eps. 
\]
Further:
\begin{align*}
u_\A(a^\i;s) &\ge u^{\UCB}_\A(a^\i; s) - 10\eps \\
&\ge u^{\LCB}_\A(\bar{a}; s) - 11\eps \\
&\ge u_\A(\bar{a}; s) - 21\eps \\
&\ge u_\A(a; s) - 22\eps. 
\end{align*}
Here, the second line follows from the fact that $a^\k$ satisfies the (approximate) incentive compatibility constraint in $\LP(a^\i)$, and the fourth line follows from the fact that $a$ and $\bar{a}$ are close and applying \cref{lemma:discretizeUtilityClose}. 
This shows $a^\i$ is indeed $\delta$-IC with $\delta = 22\eps$. 
\end{proof}

Next, let us relate the utility of the optimal contract to the contract computed by the estimated parameters. 
Let us consider the true optimal contract $s^*$, and denote the induced action under $s^*$ by $a^*$. Let us denote $\bar{a^*}$ to be the action in $\cA_D$ that is closest to $a^*$. 

\lemmaSStarFeasible*
\begin{proof}
That $s^*$ is monotone is clear. We only need to show $s^*$ satisfies the approximate-IC constraint. Fix any $a^\k\in \cA_D$. Note:
\begin{align*}
u_\A^\UCB(\bar{a^*}; s^*) &\ge u_\A(\bar{a^*}; s^*) \\
&\ge u_\A(a^*; s^*) - \eps \\
&\ge u_\A(a^\k; s^*) - \eps \\
&\ge u_\A^\LCB(a^\k; s^*) - \eps. 
\end{align*}
Here, the second line is because $a^*$ and $\bar{a^*}$ are close and applying \cref{lemma:discretizeUtilityClose}, the third line is because $a^*$ is incentive compatibly under $s^*$. 
\end{proof}

\begin{lemma}
\label{lemma:APXbound}
Fix any monotone contract $s$. $\APX(a^\i, s)\le u_\P(a^\i;s) + 2\eps$. 
\end{lemma}
\begin{proof}
\[
\APX(a^\i; s) - u_\P(a^\i;s) = \hatr(a^\i) - r(a^\i) + \sum_{j\in[m]} s_j [f_j(a^*) - \hatf_j(a^*)] 
\]
\begin{align*}
\sum_{j\in[m]} s_j [f_j(a^*) - \hatf_j(a^*)] &\le \sum_{j\in[m]} \Deltas_j [G_j(a^*) - \hatG_j(a^*)] \\
&\le \sum_{j\in[m]} \Deltas_j \cdot \eps \\
&\le \eps. 
\end{align*}
Further $\hatr(a^\i) - r(a^\i) < \eps$. 
\end{proof}

Let $a^*_D$ be the action with the highest objective among the family of $\LP(a^\i)$, and let $s^*_D$ be the solution to $\LP(a^*_D)$. 
\lemmaOPTbound*
\begin{proof}
For the first part, that $a_D^*$ is $\delta$-IC follows from \cref{lemma:deltaIC}. For the second part, 
\begin{align*}
\OPT &= r(a^*) - \sum_{j\in [m]} s^*_j f_j(a^*) \\
&\le r(\bar{a^*}) - \sum_{j\in [m]} s^*_j f_j(\bar{a^*}) + 2\eps \\
&\le \APX(\bar{a^*}; s^*) + 4\eps \\
&\le \APX(a^*_D; s^*_D) + 4\eps \\
&\le u_\P(a^*_D; s^*_D) + 6\eps. 
\end{align*}
Here, the second line follows from the fact that $a^*$ and $\bar{a^*}$ are close, the third line follows from the previous \cref{lemma:APXbound}, the fourth line follows from the fact that $s_D^*$ is optimal among the family of LPs using estimated parameters, and the final line is from \cref{lemma:APXbound} again. 
\end{proof}


Now, we apply the very interesting observation in \cite{dutting2021complexity}, which relates the utility of $\delta$-IC contracts to that of exact IC contracts. For completeness, we include the proof (phrased in a slightly different way from their original proof). 

\begin{lemma}[\cite{dutting2021complexity}]
Let $(s,a)$ be a $\delta$-IC contract. Then $s' := (1-\sqrt{\delta})s + \sqrt{\delta} \pi$ achieves utility at least $(1-\sqrt{\delta}) u_\P(a; s) - (\sqrt{\delta} - \delta)$. 
\end{lemma}
\begin{proof}
Denote $p(a; s)$ to be the expected payment under $s$ if the agent were to take action $a$. 
It is easy to calculate
\[
u_\P(a; s') = (1-\sqrt{\delta}) u_\P(a; s)
\]
Now consider the contract $s'$, and suppose this contract implements $a'$, then IC is satisfied:
\[
p(a'; s') - c(a') \ge p(a; s') - c(a),
\]
a direct calculation gives:
\[
(1-\sqrt{\delta}) p(a'; s) + \sqrt{\delta} r(a') - c(a') \ge (1-\sqrt{\delta}) p(a; s) + \sqrt{\delta} r(a) - c(a). 
\]
Simplifying further:
\[
p(a';s) - c(a') + \sqrt{\delta}[ r(a') - p(a';s) ] \ge p(a;s) - c(a) + \sqrt{\delta}[ r(a) - p(a;s) ]
\]

Since $a$ is $\delta$-IC under $s$, then
\[
p(a; s) - c(a) \ge p(a' ; s) - c(a') - \delta. 
\]
The above two inequalities directly gives:
\[
\sqrt{\delta}[ r(a') - p(a';s) ] \ge \sqrt{\delta}[ r(a) - p(a;s) ] - \delta
\]

The principal's utility under the exact IC contract $s'$ can be lower bounded:
\begin{align*}
u_\P(a' ; s') &= (1 - \sqrt{\delta}) (r(a') - p(a'; s)) \\
&\ge (1 - \sqrt{\delta}) (r(a) - p(a; s) - \sqrt{\delta}) \\
&\ge (1 - \sqrt{\delta}) (r(a) - p(a; s)) - (\sqrt{\delta} - \delta). 
\end{align*}
This completes the proof. 
\end{proof}

\begin{proposition}
The principal can find a $6\eps + 10\sqrt{\eps}$ approximately optimal contract using $\widetilde{O}(1/\eps^3)$ samples. 
\end{proposition}
\begin{proof}
Recall the contract $(s_D^*, a_D^*)$ is $22\eps$-IC and the utility is at least a $6\eps$ approximation (\Cref{prop:approxICandOPT}). Then, using the above result, $s_D^*$ can be converted into an exact IC contract that achieves utility at least:
\begin{align*}
(1 - \sqrt{22\eps}) (u_\P(a^*; s^*) - 6\eps) - (\sqrt{22\eps} - 22\eps) &\ge u_\P(a^*; s^*) - \sqrt{22\eps}u_\P(a^*; s^*) - 6\eps - \sqrt{22\eps} \\
&\ge u_\P(a^*; s^*) - 6\eps - 10\sqrt{\eps}. 
\end{align*}
\end{proof}

\begin{proof} [Proof of \Cref{thm:sample}]
Follows directly from the above proposition. 
\end{proof}

\szdelete{
\subsection{Can assume $\beta = 1$ induce maximum action}
Suppose $\beta = 1$ induces the action $a < E$. Let $\beta = 1$ induce the action $a$. We show any action $b > a$ will be too costly to induce. Let us suppose the contract $s$, induces $b$. Then
\[
\sum_j s_j f_j(b) - c(b) > \sum_j s_j f_j(a) - c(a). 
\]
Also, since the linear contract $\beta = 1$ induces $a$, then
\[
\sum_j r(a) - c(a) > \sum_j r(b) - c(b). 
\]

Then:
\[
\sum_j (\pi_j - s_j) f_j(a) > \sum_j (\pi_j - s_j) f_j(b)
\]
So the question is, does there exist a positive sequence such that:
\[
\sum_j y_j (f_j(a) - f_j(b)) > 0
\]

\subsubsection{Idea 2}
Let us suppose that contract $s$ induce $b$. Then:
\[
\sum_j \Deltas_j G_j(b) - c(b) > \sum_j \Deltas_j G_j(a) - c(a)
\]
Also:
\[
r(a) - c(a) > r(b) - c(b)
\]

The utility of inducing $b$ is:
\begin{align*}
r(b) - \sum_j \Deltas_j G_j(b) < c(b) - c(a) - \sum_j \Deltas_j G_j(a)
\end{align*}

\subsubsection{Idea 3}
By first order condition. 
\begin{align*}
\sum_j s_j G'_j(b) &= c'(b) \\
\sum_j \pi_j G'_j(b) &< c'(b)
\end{align*}
Second line comes from $r'(b) < c'(b)$. 
}

\end{document}